\documentclass[a4paper,11pt,oneside]{article}
\usepackage[margin=1in]{geometry}
\usepackage[utf8]{inputenc}
\usepackage{amsmath}
\usepackage{amsfonts}
\usepackage{amsthm}

\newtheorem{theorem}{Theorem}[section]
\newtheorem{lemma}[theorem]{Lemma}
\usepackage{bbm}
\usepackage{enumitem}
\usepackage{algorithm}
\usepackage{algorithmic}
\usepackage{graphicx}
\graphicspath{ {./graphs/} }
\usepackage{mathtools}
\usepackage{caption}
\usepackage{subcaption}

\usepackage{cite}
\usepackage[font=scriptsize]{caption}
\pdfoutput=1

\title{ExtendedHyperLogLog: Analysis of a new Cardinality Estimator}

\author{Tal Ohayon \\ tohayonh@gmail.com}

\begin{document}
\date{}
\maketitle
\date{\vspace{-5ex}}
\begin{abstract}
  We discuss the problem of counting distinct elements in a stream. A stream is usually considered as a sequence of elements that come one at a time. An exact solution to the problem requires memory space of the size of the stream. For many applications this solution is infeasible due to very large streams. The solution in that case, is to use a probabilistic data structure (also called sketch), from which we can estimate with high accuracy the cardinality of the stream. We present a new algorithm, ExtendedHyperLogLog (EHLL), which is based on the state-of-the-art algorithm, HyperLogLog (HLL). In order to achieve the same accuracy as HLL, EHLL uses $16\%$ less memory. In recent years, a martingale approach has bean developed. In the martingale setting we receive better accuracy at the price of not being able to merge sketches. EHLL also works in the martingale setting. Martingale EHLL achieves the same accuracy as Martingale HLL using $12\%$ less memory.
\end{abstract}

\section{Introduction}
The problem of counting the number of distinct elements in a stream (also called cardinality), pioneered by Flajolet and Martin in the seminal article from 1985 \cite{DBLP:journals/jcss/FlajoletM85}. The stream consists of a sequence of elements, usually large, so we cannot store it. We get the elements one by one, and would like to be able to count, at any moment, the number of distinct elements that have arrived so far.

The problem of cardinality estimation is very practical and appears in a wide range of industrial applications. Let us mention several such applications. In the field of databases it is used for query optimization and database design \cite{DBLP:journals/jcss/FlajoletM85}. Networking applications compute the number of distinct flows per IP to detect denial-of-service attacks \cite{Estan2006BitmapAF}. A web search engine may compute the number of distinct users that have used a certain word in a query \cite{10.1145/1353343.1353418}.

Most of the algorithms for cardinality estimation consist of a sketch and an estimator. A \emph{sketch} is a probabilistic data structure that keeps some of the information of the stream. The estimator gives an approximation of the cardinality. 
For example, the FM85 sketch, of Flajolet and Martin \cite{DBLP:journals/jcss/FlajoletM85}, is a bit map. For every element that arrives, we take the index of the least significant bit of its hashed value, and fill $'1'$ in the corresponding place in the bitmap. The cells of the bitmap are sometimes called coupons. The estimator of FM85 uses the location of the leftmost $'0'$ to estimate the cardinality.

In recent years, a new kind of algorithms has been developed. These algorithms work in a sequential setting. They use the information from the sequence of sketches, created during the stream processing \cite{DBLP:conf/kdd/Ting14, Cohen}. This method results in much more accurate algorithms, but comes at a price, it cannot be distributed. Thus, we cannot merge sketches from different streams (we will refer to these sketches as \emph{non-mergeable} sketches \cite{pettie2021nonmergeable}). It is common that, when we adjust a sketch from the distributed setting to the sequential setting with the technique of \cite{DBLP:conf/kdd/Ting14, Cohen}, we add the prefix martingale (explained in Section 5).

Another interesting approach is to leave the runtime performance aside and concentrate only on memory consumption. Pettie et al.\ \cite{pettie2020information} designed a new sketch called Fishmonger. Fishmonger is based on a compressed version of FM85. It uses arithmetic coding, which results in linear time to encode and decode the sketch between updates. Thus, Fishmonger may be impractical for most applications.

The state-of-the-art algorithm in the distributed setting, with O(1) update time, is the HyperLogLog (HLL) sketch \cite{Flajolet2007HyperLogLogTA}. HLL keeps only the location of the rightmost $'1'$ in FM85 sketch. In \cite{10.1109/TNET.2020.2970860}, a new algorithm has been suggested, HyperLogLog-TailCut (HLL-TC). HLL-TC uses the minimal counter value of HLL as a base counter, and keeps in each of the other counters the offset from the base counter. Using this technique, \cite{10.1109/TNET.2020.2970860} argued that allocating 4 bits per counter is suffice. In that way the memory usage of HyperLogLog has been reduced.

The problem with this technique is that there may be an offset that requires a memory of more than 4 bits. In that situation, it was suggested to use the maximum value that a counter can store. This property causes the estimation to be dependant on the order of the elements in the stream. Moreover, merging a few sketches can result in an estimation that is different than what would have been received if one sketch was used. For some applications this may not be a problem. In that case we suggest to use the technique of HLL-TC in our algorithm EHLL, named for EHLL-TC. We note that, in order to achieve the same accuracy as HLL-TC, EHLL-TC uses 10\% less memory.

When comparing different algorithms one should take into account variance and memory. We usually consider the relative variance, which is the ratio between the variance and the square of the true cardinality. To reduce the variance, most algorithms make duplicates of the sketch. Namely, they split uniformly the stream into sub-streams, and direct each sub-stream to its corresponding sketch. This method is known as stochastic averaging \cite{DBLP:journals/jcss/FlajoletM85} (We will explain in more details in Section 3). When using $m$ duplicates of the sketch we reduce the variance by a factor of $m$, but increase the memory space by a factor of $m$. Thus, as suggested in \cite{pettie2021nonmergeable}, a fair way to compare sketches is to take their memory variance product (MVP). The MVP of the mergeable sketches that mentioned above described in Table 1. 

\begin{table*}
\centering
\begin{tabular}{ |p{4cm}|p{4.5cm}|p{4cm}|  }
 \hline
 Sketch & MVP & Comments\\
 \hline
 PCSA \cite{DBLP:journals/jcss/FlajoletM85}  & $0.6 \log{U} \approx 38.9$    &\\
 LogLog \cite{loglog} & $1.69 \log{\log{U}} \approx 10.11$ & \\
 HyperLogLog \cite{Flajolet2007HyperLogLogTA}& $1.08 \log{\log{U}} \approx 6.48$ & \\
 ExtendedHyperLogLog & $0.78 (\log{\log{U}} + 1) \approx 5.46$ & New\\
 Fishmonger \cite{pettie2020information}& $1.98$ & Not practical due to updates run time \\
 \hline
\end{tabular}
\caption{MVP of mergeable sketches. Cardinality assumed to be up to $U = 2^{64}$.}
\label{table:1}
\end{table*}

We start by describing the algorithms, FM85, HLL and our EHLL. Next, we dive into a theoretical analysis of EHLL. We provide the martingale setting and its error analysis. Finally, we present simulation results comparing the estimators.  
\section{Notations}
Our algorithm is based on the FM85 sketch \cite{DBLP:journals/jcss/FlajoletM85} and  the HyperLogLog sketch \cite{Flajolet2007HyperLogLogTA}. Thus, we will keep the notations used in these sketches. We denote by $n$ the true cardinality that we want to estimate. Flajolet et al.\ \cite{DBLP:journals/jcss/FlajoletM85} described a method called stochastic averaging. In this method, we split the stream uniformly into smaller streams. We denote by $m$ the number of smaller streams.

In the case of small cardinalities, where the estimation of HLL is less than $\frac{5}{2} m$, another estimation technique is used, based on the sketch that has been constructed, LinearCounting \cite{LinearCounting}. In EHLL, we deal with this case in a similar way.

\section{Algorithms Description}
\subsection{PCSA}
 The Probabilistic Counting with Stochastic Averaging (PCSA) algorithm is based on FM85 sketch \cite{DBLP:journals/jcss/FlajoletM85}. The sketch is a bitmap that initialized with zeroes. For every new element that arrives, we sample from $\text{Geo}(1/2)$ distribution. We use this sampled value as the location to fill $'1'$ in the bitmap. Flajolet and Martin implemented this procedure as follows. For every element they used it's hashed value, this way they could get uniform distribution above binary words. Then, they used the location of the least significant bit in the binary string, denote it as $\rho$. Note that $\rho \sim \text{Geo}(1/2)$. Also note that if an element arrived in the second time, or more, it doesn't affect the sketch. This is very important feature of the sketch because we don't want it to be sensitive to duplicates.
\begin{algorithm}
\caption{FM sketch}\label{euclid}
\begin{algorithmic}[1]
\FOR{\texttt{$i = 0$ to $L - 1$}}
    \STATE $\text{BITMAP}[i] \gets 0$
\ENDFOR
\FOR{\texttt{$v$ in $M$}}
    \STATE $\text{index} \gets \rho (\text{hash}(v))$
    \IF{$\text{BITMAP}[\text{index}] = 0$}
        \STATE $\text{BITMAP}[\text{index}] \gets 1$
    \ENDIF
\ENDFOR
\end{algorithmic}
\end{algorithm}

Flajolet and Martin propose to use the left most zero in the bitmap as an indicator of $\log_2{n}$. Denote this value as $R$, they showed that $\mathbb{E}[R] \approx \log_2{\phi n}, \hspace{0.1cm} \phi \approx 0.77351$. The approximation of $n$ using $R$ yields typical error of one binary order which is too high for many applications \cite{DBLP:journals/jcss/FlajoletM85}.   

In order to reduce the error, they suggest to duplicate the sketch into $m$ sketches. When a new element arrived, one of the bitmaps chosen randomly and updated in the same way as in Algorithm 1. They implemented it the following way. They used the first $\log_2{m}$ bits of the element hashed value, as the address of the bitmap to be updated (this way the same element cannot be inserted to two different bitmaps). Then, they used the other bits to perform Algorithm 1. The cardinality is approximated with the average $R$ values of the $m$ bitmaps. Flajolet and Martin called this procedure stochastic averaging. With this procedure the relative error reduced to $0.78/ \sqrt{m}$.

\subsection{HyperLogLog}
HyperLogLog \cite{Flajolet2007HyperLogLogTA} keeps an array of small counters of size 6 bits (in order to estimate cardinalities up to $2^{64}$). For every element that arrives, HyperLogLog chooses randomly a cell in the array and sample from a  $\text{Geo}(1/2)$ distribution. If the sampled number is bigger than the current value in the cell, then the cell is updated with the sampled value. The way of choosing the cell and generating the $\text{Geo}(1/2)$ - variable is described in Algorithm 2 (same as PCSA). 

Let $v$ be the next new element in the stream. HyperLogLog chooses randomly a cell, so it has a probability of $\frac{1}{m}$ to fall into cell $j$. Given that it felt into cell $j$, the probability that cell $j$ will be updated is $2^{-C[j]}$, which is the probability of generating a $\text{Geo}(1/2)$ - variable that is bigger than the current value in cell $j$. Thus, the probability that the sketch will be updated is $\frac{1}{m} \left( \sum_{j = 1}^m 2^{-C[j]}\right)$.
HyperLogLog uses the indicator $Z = \left( \sum_{j = 1}^m 2^{-C[j]}\right)^{-1}$ in order to estimate the cardinality. Note that, in the arrival of a new element, the probability that the sketch will be updated equals $\frac{1}{m Z}$. Flajolet et al.\ \cite{Flajolet2007HyperLogLogTA} calculated $\mathbb{E}[Z]$ in order to get the bias correction constant $\alpha_m$. They also showed that its relative error is $1.04/\sqrt{m}$.
\begin{algorithm}
\caption{HyperLogLog}\label{euclid1}
\begin{algorithmic}[1]
\STATE $C \gets$ array of size $m = 2^b$, where $b \in \mathbb{N}$, initialized with zeroes.
\FOR{$v$ in $M$}
    \STATE $x \gets hash(v)$
    \STATE $j \gets$ first $b$ bits of $x$
    \STATE $y \gets$ the rest of the bits of $x$
    \STATE $C[j] \gets \max \{C[j], \rho(y)\}$, where $\rho$ gives the location of the least significant bit.
\ENDFOR
\STATE $Z \gets$ $\left( \sum_{j = 1}^m 2^{-C[j]}\right)^{-1}$
\RETURN $\alpha_m m^2 Z$ 
\end{algorithmic}
\end{algorithm}

Note that each counter saves the location of the right most $'1'$ in FM85 sketch. The location of the right most $'1'$ is less informative than the location of the left most $'0'$. For the left most $'0'$ we have to save the whole bitmap, while for the right most $'1'$ we always keep the maximum value of a $\text{Geo}(1/2)$ sample. In order to save the location of the left most $'0'$ we need $O(\log{n})$ bits, and to save the location of the right most $'1'$ we need $O(\log(\log{n}))$ bits. Thus, when the two sketches have the same amount of memory, HyperLogLog achieves better accuracy.

\subsection{ExtendedHyperLogLog}
The ExtendedHyperLogLog algorithm may be seen as a compromise between HyperLogLog and FM85. It keeps the information about the location of the right most $'1'$ ($C_1$, array of integers), such as HyperLogLog. The additional information from FM85, is the state of the location, left to the right most $'1'$ in FM85 sketch ($C_2$, array of bits). For example, let $j \in [m]$, and assume that the location of the right most $'1'$ in FM85 sketch is $k$, then $C_1[j] = k$. If, at least one of the $\text{Geo}(1/2)$ observations equals $k - 1$, then $C_2[j] = 1$. Otherwise, it will be zero. The updates of the ExtendedHyperLogLog sketch describe in Algorithm 3.

The estimation function works in a similar way as in HyperLogLog. Consider the moment that we want to estimate the number of distinct elements. We will use the probability that the sketch will be changed in the arrival of a new element. In that case, the probability that the element felt into cell $j$ is $\frac{1}{m}$ and it has probability of $2^{-C_1[j]}$ to be greater than $C_1[j]$. If $C_2[j] = 1$, then the probability that this row will be changed is $2^{-C_1[j]}$. If $C_2[j] = 0$, then we have to also include the probability that the value generated by the new element equals $C_1[j] - 1$. Because this case also changes the sketch. This occurs with probability $2^{- C_1[j] + 1}$. Overall, we got that the probability that the sketch will be changed is $q = \frac{1}{m} \left( \sum_{j = 1}^m 2^{-C_1[j]} + (1 - C_2[j])2^{-C_1[j] + 1}\right)$. We define  $Y = \frac{1}{m q}$. In the next section we will analyse the expectation and variance of $Y$, in order to get the bias correction constant $\gamma_m$.

\begin{algorithm}
\caption{ExtendedHyperLogLog}\label{euclid3}
\begin{algorithmic}[1]
\STATE $C_1 \gets$ array of size $m = 2^b$ where $b \in \mathbb{N}$, initialized with zeroes.
\STATE $C_2 \gets$ bitmap array of size $m = 2^b$, initialized with ones.
\FOR{\texttt{$v$ in $M$}}
    \STATE $x \gets hash(v)$
    \STATE $j \gets$ first $b$ bits of $x$
    \STATE $y \gets$ the rest bits of $x$
    \IF{$\rho(y) =  C_1[j] + 1$} 
        \STATE $C_1[j] \gets \rho(y)$
        \STATE $C_2[j] \gets 1$
    \ELSIF{$\rho(y) >  C_1[j] + 1$} 
        \STATE $C_1[j] \gets \rho(y)$
        \STATE $C_2[j] \gets 0$
    \ELSIF{$\rho(y) =  C_1[j] - 1$ \AND $C_2[j] = 0$} 
        \STATE $C_2[j] \gets 1$
    \ENDIF
\ENDFOR
\STATE $Y \gets$ $\left( \sum_{j = 1}^m 2^{-C_1[j]} + (1 - C_2[j])2^{-C_1[j] + 1}\right)^{-1}$
\RETURN $\gamma_m m^2 Y$ 
\end{algorithmic}
\end{algorithm}

\newpage
\section{Mean and Variance Analysis}
As mentioned in the previous section, the variable which we estimate by the number of distinct elements is $Y$, where $\frac{1}{m Y}$ is the probability that the sketch will be updated in the next arrival of a new element. In the case of HLL we will refer to this variable as $Z$.

All left to do before writing the exact expression for the expectation is to calculate the probability that in cell $j$, in HLL sketch, the number is $k$, given that $n_j$ distinct elements arrived at this cell. Since the number in the cell is the maximum of $n_j$ random variables $\text{Geo}(1/2)$ we get that $P(C[j] = k) = \left(1 - \frac{1}{2^k}\right)^{n_j} - \left(1 - \frac{1}{2^{k - 1}}\right)^{n_j}$, for $k \geq 1$.
 
Flajolet et al.\cite{Flajolet2007HyperLogLogTA} wrote the exact expression for the expectation of $Z$ in the following way
\begin{equation}
\mathbb{E}[Z] = \sum_{k_1,...,k_m \geq 1} \frac{1}{\sum_{j = 1}^m 2^{-k_j}} \sum_{n_1 + ...+n_m = n} \binom{n}{n_1, ...,n_m} \frac{1}{m^n} \prod_{j = 1}^m \gamma_{n_j,k_j}
\end{equation}
where
\begin{equation}
\gamma_{n_j,k_j} = \left(1 - \frac{1}{2^{k_j}}\right)^{n_j} - \left(1 - \frac{1}{2^{k_j - 1}}\right)^{n_j}.
\end{equation}
\\

In the case of EHLL, we have the addition of the coupon to the left of the right most coupon that has been collected, which yields
\begin{equation} \label{eq: Ey}
\begin{split}
    \mathbb{E}[Y] &= \sum_{x \in \{0, 1\}^m} \sum_{k_1,...,k_m \geq 1} \left[\frac{1}{\sum_{j = 1}^m 2^{-k_j} + (1 - x_j)2^{-k_j + 1}} \right. \\
    & \hspace{3cm} \left.\sum_{n_1 + ...+n_m = n} \binom{n}{n_1, ...,n_m} \frac{1}{m^n} \prod_{j = 1}^m \phi_{n_j,k_j, x_j} \right]
\end{split}
\end{equation}
where
\begin{equation} \label{eq: phi}
\phi_{n_j,k_j, x_j} = \gamma_{n_j,k_j} \left(x_j + (-1)^{x_j} \left(1 - \frac{1}{2^{k_j - 1}}\right)^{n_j}\right).
\end{equation}

We know that the expectation of $Y$ doesn't equal to the cardinality of the stream. We do know that the expectation depends on the cardinality $n$ and on the number of duplicates of the sketch, which is $m$. In the following we will calculate this relation. After having a closed form of these variables we will be able to calculate the cardinality based on the value of $Y$ at any moment of the stream.
\begin{theorem}
The variable $Y$ has the following behaviour as $n \rightarrow \infty$
\begin{equation} \label{eq:theorem4}
\mathbb{E}_n[Y] = \frac{n}{m} \left( \int_0^\infty f(u)^m du + \frac{\epsilon}{m} + o(1)\right)
\end{equation}
where $f(u) = \log_2{\frac{2 + u}{1 + u}} + \log_2{\frac{4 + 3u}{3 + 3u}} - \log_2{\frac{4 + u}{3 + u}}$
and $|\epsilon| \leq 8.4 \cdot 10^{-5}$.

\end{theorem}

Due to the similarity of the sketches, many of the steps and the mathematical tools in the analysis of $Z$ \cite{Flajolet2007HyperLogLogTA} can be used to analyze $Y$. Thus, in order to make this paper self contained, we repeat these steps with some little modifications that relates to the differences between $Z$ and $Y$. The analysis consists mathematical tools such as poissonization, Mellin transform and Laplace integrals. We will review these tools through the paper. The proof consists of 3 main steps:
\begin{enumerate}
    \item Use poissonization to simplify the exact expression of the expectation (\ref{eq: Ey}). This means that we analyze the variable $Y$ under the assumption that $n \sim Pois(\lambda)$. We will get that $E_{\mathcal{P}(\lambda)}[Y] = \int_0^\infty G(\lambda,t)^m dt$, where $G(\lambda,t)$ is a sum of elements. 
    \item Approximate the integrand. We can't calculate directly the sum $G(\lambda,t)$. So we define a new sum $G_2(\lambda, t)$.
    We define $\lambda = mx$ and $t = ux$. Then, we use Mellin transform to get that $G_2(mx, ux) \approx f(u)$.  
    \item Calculate the error which results from using $f(u)$ instead of $G(mx, ux)$.
\end{enumerate}

\subsection{Poisson Model}
In the poissonization process we calculate the expectation with the assumption that the number of distinct elements distributed $\text{Pois}(\lambda)$. Denote this sum as $\mathbb{E}_{\mathcal{P}(\lambda)}[Y]$. When taking $\lambda = n$ we can get that (the proof similar to Proposition 3 in \cite{Flajolet2007HyperLogLogTA})
\begin{equation} \label{eq:7}
    \mathbb{E}_n(Y) = \mathbb{E}_{\mathcal{P}(n)}[Y] + O(1)
\end{equation}
Poissonization enable us to use Mellin transform in the next steps of the analysis \cite{loglog}.

Under The Poisson model
\begin{equation}
  \mathbb{E}_{\mathcal{P}(\lambda)}[Y] = \sum_{n \geq 0} E_n[Y] e^{-\lambda}\frac{\lambda^n}{n!}  
\end{equation}
A few calculations yields
\begin{align*}
    &\sum_{n \geq 0} \sum_{n_1 + ...+n_m = n}e^{-\lambda} \frac{\lambda^n}{n!} \binom{n}{n_1, ...,n_m} \frac{1}{m^n} \prod_{j = 1}^m \phi_{n_j,k_j, x_j} \\
    & \hspace{0.5cm}= \sum_{n \geq 0} \sum_{n_1 + ...+n_m = n} \prod_{j = 1}^m e^{-\frac{\lambda}{m}} \frac{\left(\frac{\lambda}{m}\right)^{n_j} \phi_{n_j,k_j, x_j}}{n_j!} \\
    & \hspace{0.5cm} = \prod_{j = 1}^m \left[e^{-\frac{\lambda}{m}} \sum_{n_j \geq 0} \frac{\left(\frac{\lambda}{m}\right)^{n_j} \phi_{n_j,k_j, x_j}}{n_j!} \right] \\
    & \hspace{0.5cm} = \prod_{j = 1}^m g(x_j, k_j, \lambda)
\end{align*}
where
\begin{equation}
\begin{split}
    g(x_j, k_j, \lambda) &= x_j\left(e^{-\frac{\lambda}{m2^{k_j}}} - e^{-\frac{\lambda}{m2^{k_j - 1}}}\right) \\
    &\hspace{1cm}+ (-1)^{x_j}\left(e^{-\frac{\lambda}{m}\left(\frac{3}{2^{k_j}} - \frac{2}{2^{2k_j}}\right)} - e^{-\frac{\lambda}{m}\left(\frac{4}{2^{k_j}} - \frac{4}{2^{2k_j}}\right)}\right).
\end{split}
\end{equation}
After using the above we get
\begin{equation} \label{eq:12}
\begin{split}
      \mathbb{E}_{\mathcal{P}(\lambda)}[Y] = \sum_{x \in \{0, 1\}^m} \sum_{k_1,...,k_m \geq 1} \frac{1}{\sum_{j = 1}^m 2^{-k_j} + (1 - x_j)2^{-k_j + 1}} \prod_{j = 1}^m g(x_j, k_j, \lambda),
\end{split}
\end{equation}

In order to continue with the simplification of (\ref{eq:12}) we use the identity
\begin{equation}
  \frac{1}{a} = \int_0^\infty e^{-at}dt.  
\end{equation}
We get this separation of variables
\begin{align*}
    \mathbb{E}_{\mathcal{P}(\lambda)}[Y] &= \sum_{x \in \{0, 1\}^m} \sum_{k_1,...,k_m \geq 1} \int_0^\infty \left( e^{-t(\sum_{j = 1}^m 2^{-k_j} + (1 - x_j)2^{-k_j + 1})} \prod_{j = 1}^m g(x_j, k_j, \lambda) \right) dt \\
    &= \int_0^\infty G(\lambda, t)^m dt,
\end{align*}
Where
\begin{align*}
G(\lambda, t) &= \sum_{k \geq 1} \left[\left(e^{-\frac{\lambda}{m2^{k}}} - e^{-\frac{\lambda}{m2^{k - 1}}}\right)e^{-t2^{-k}} \right. \\
&\hspace{1cm}+ \left. \left(e^{-\frac{\lambda}{m}\left(\frac{3}{2^{k}} - \frac{2}{2^{2k}}\right)} - e^{-\frac{\lambda}{m}\left(\frac{4}{2^{k}} - \frac{4}{2^{2k}}\right)}\right) (e^{-3t2^{-k}}-e^{-t2^{-k}}) \right]. 
\end{align*}
For convenience we substitute $\lambda = mx$ to obtain
\begin{equation}
  \mathbb{E}_{\mathcal{P}(mx)}[Y] = \int_0^\infty G(mx, t)^m dt.  
\end{equation}
We will make another change of variables $t = xu$. Then
\begin{equation} \label{eq:13}
    \mathbb{E}_{\mathcal{P}(x)}[Y] = H\left(\frac{x}{m}\right), \hspace{1cm} H(x) = x\int_0^\infty G(mx, ux)^m du,
\end{equation}
while
\begin{equation}
\begin{split} \label{eq:G(mx, ux)}
G(mx, ux) &= \sum_{k \geq 1} \biggl[\left(e^{-\frac{x}{2^{k}}} - e^{-\frac{x}{2^{k - 1}}}\right)e^{-ux2^{-k}} \\
& \hspace{1cm} + \left(e^{-x\left(\frac{3}{2^{k}} - \frac{2}{2^{2k}}\right)} - e^{-x\left(\frac{4}{2^{k}} - \frac{4}{2^{2k}}\right)}\right) \\
& \hspace{1cm} \cdot (e^{-3ux2^{-k}}-e^{-ux2^{-k}}) \biggr].
\end{split}
\end{equation}

\subsection{Approximation of the Integrand}
In order to calculate the sum $G(mx, ux)$, we use the Mellin transform. 
Mellin transform defined as
\begin{equation}
    \mathcal{M}[f(x); s] = f^*(s) = \int_0^\infty f(x)x^{s - 1} dx.
\end{equation}
We would like to use the property that the transform, for the so called harmonic sums \cite{Flajolet1985SomeUO}
\begin{equation}
  F(x) = \sum_k \lambda_k f(\mu_k x), 
\end{equation}
factorizes as
\begin{equation}
  F^*(s) = \left( \sum_k \lambda_k \mu_k^{-s} \right)f^*(s).  
\end{equation}

To this end we need to write the terms in the sum on the right hand side of (\ref{eq:G(mx, ux)}) as a function $f(\mu_k x)$. Unfortunately we have terms of the shapes $\frac{x}{2^k}$ and $\frac{x}{2^{2k}}$. We omit the terms $\frac{x}{2^{2k}}$, and in the next section we will prove that their contribution can be neglected.
Denote as 
\begin{equation} \label{eq:20}
\begin{split}
G_2(mx, ux) &= \sum_{k \geq 1} \left[\left(e^{-\frac{x}{2^{k}}} - e^{-\frac{x}{2^{k - 1}}}\right)e^{-ux2^{-k}} \right.\\
&\hspace{1cm} + \left. \left(e^{-\frac{3x}{2^{k}}} - e^{-\frac{4}{2^{k}}}\right)  (e^{-3ux2^{-k}}-e^{-ux2^{-k}}) \right].
\end{split}
\end{equation}

Using Lemma 1 from \cite{Flajolet2007HyperLogLogTA}, which is the analysis of the first term on the right hand side of (\ref{eq:20}), that uses the Mellin transform, we get the following lemma.
\begin{lemma}
For arbitrary fixed $u > 0$, the function $G_2(mx, ux)$ has the asymptotic behaviour
\begin{equation}
    G_2(mx, ux) = 
\begin{cases}  
    f(u)(1 + O(x^{-1})) + u\epsilon,& u \leq 1, \\
    f(u)(1 + \epsilon + O(x^{-1})),& u > 1,
\end{cases}    
\end{equation}
as $x \rightarrow \infty$, where
\begin{equation}
f(u) = \log_2{\frac{2 + u}{1 + u}} + \log_2{\frac{4 + 3u}{3 + 3u}} - \log_2{\frac{4 + u}{3 + u}},
\end{equation}
and $|\epsilon| \leq 2.1 \cdot 10^{-5}$.
\end{lemma}

\subsection{Integrand Error Analysis}
All left to do is to bound the gap between $\int_0^\infty f(u)^m du$\\ and $\int_0^\infty G(mx, ux)^m du$. 
\begin{lemma}
The function $\int_0^\infty G(mx, ux) du$ has the following asymptotic behavior as $x \rightarrow \infty$:
\begin{equation}
\int_0^\infty G(mx, ux)^m du = \int_0^\infty f(u)^m du + \frac{\epsilon}{m} + o(1)
\end{equation}
where $|\epsilon| \leq 2.1 \cdot 10^{-5}$.
\end{lemma}

\begin{proof}
The proof consists of two parts. First we will bound the contribution of the terms of the form $\frac{x}{2^{2k}}$ in $G(mx, ux)$. In other words we will bound the gap between $G(mx, ux)$ and $G_2(mx, ux)$. Second, we will bound the gap between $\int_0^\infty G(mx, ux)^m du$ and $\int_0^\infty f(u)^m du$.

We can write the gap between $G(mx, ux)$ and $G_2(mx, ux)$ as  

\begin{align*}
S &= \sum_{k = 1}^\infty \left[\left(e^{-\frac{3x}{2^{k}}} - e^{-x\left(\frac{3}{2^{k}} - \frac{2}{2^{2k}}\right)} \right) (e^{-3ux2^{-k}}-e^{-ux2^{-k}}) \right. \\
&\hspace{1cm} + \left. \left(e^{-x\left(\frac{4}{2^{k}} - \frac{4}{2^{2k}}\right)} - e^{-\frac{4x}{2^{k}}}\right) (e^{-3ux2^{-k}}-e^{-ux2^{-k}}) \right]\\
&= \sum_{k = 1}^\infty \left[e^{-\frac{3x}{2^k}}(1 - e^{\frac{2x}{2^{2k}}}) (e^{-3ux2^{-k}}-e^{-ux2^{-k}}) \right. \\
 &\hspace{1cm} \left. + e^{-\frac{4x}{2^k}}(e^{\frac{4x}{2^{2k}}} - 1) (e^{-3ux2^{-k}}-e^{-ux2^{-k}}) \right]\\
&\equiv S_1 + S_2
\end{align*}

We will start the analysis of $S_1$ by splitting it into two sub-sums.
\begin{align*}
    S_1 &= \sum_{k < \frac{1}{2}\log_2{2x}} e^{-\frac{3x}{2^k}}(e^{\frac{2x}{2^{2k}}} - 1) (e^{-ux2^{-k}}-e^{-3ux2^{-k}})\\
    & \hspace{1cm} + \sum_{k \geq \frac{1}{2}\log_2{2x}} e^{-\frac{3x}{2^k}}( e^{\frac{2x}{2^{2k}}} - 1) (e^{-ux2^{-k}}-e^{-3ux2^{-k}}) \\
    &\equiv S_{1,1} + S_{1,2}.
\end{align*}
For arbitrary fixed $x > 1$, $u \geq 0$ and $k \in [1,\frac{1}{2}\log_2{2x}]$, the functions $f(k) = e^{-\frac{3x}{2^k}}(e^{\frac{2x}{2^{2k}}} - 1)$ and $e^{-ux2^{-k}}$ are monotonically increasing (see Lemma \ref{bounding f(k) lemma} in the appendix). $1 - e^{-x} < 1$ for all $x \geq 0$. Thus, we can conclude that
\[S_{1,1} < \frac{1}{2}(e - 1)\log_2{(2x)} e^{-\sqrt{\frac{x}{2}}(3 + u)}\]
For $S_{1,2}$ we use the inequality $e^x - 1 < \frac{7x}{4}$ for $0 < x < 1$ to get
\begin{align*}
    S_{1,2} &\leq  \sum_{k \geq \frac{1}{2}\log_2{2x}} \frac{7}{2} \frac{x}{2^{2k}}e^{-\frac{3x}{2^k}} (e^{-ux2^{-k}}-e^{-3ux2^{-k}}) \\
    &\leq  \sum_{k \geq 1} \frac{7}{2} \frac{x}{2^{2k}}e^{-\frac{3x}{2^k}} (e^{-ux2^{-k}}-e^{-3ux2^{-k}}) \\
    &= \frac{7}{2x} \sum_{k \geq 1} \frac{x^2}{2^{2k}}e^{-\frac{3x}{2^k}} (e^{-ux2^{-k}}-e^{-3ux2^{-k}}).
\end{align*}
Using Lemma \ref{bounding s lemma} (see appendix) we get that
\begin{equation}
S_{1,2} \leq 
 \frac{7}{2x} g(u) \left(1 + \epsilon + O\left(\frac{1}{x} \right)\right)
\end{equation}
where $g(u) = \frac{1}{\log{2}} \left( \frac{1}{(3 + u)^2} - \frac{1}{(3 + 3u)^2} \right)$ and $|\epsilon| \leq 1.5 \cdot 10^{-5}$.

Note that for $u > 0$ and $x \rightarrow \infty$, $S_{1,1} < S_{1,2}$. The same calculation is valid for the analysis of $S_2$, and yields a smaller bound than what we got for $S_1$. Recall that $G(mx, ux) = G_2(mx, ux) + S$, using Lemma 4.2 we get
\begin{equation}
G(mx, ux) = 
\begin{cases}  
    f(u) + u\epsilon + O\left(\frac{1}{x} \right)(f(u) + g(u)),& u \leq 1, \\
    f(u)(1 + \epsilon) + O\left(\frac{1}{x} \right)(f(u) + g(u)),& u > 1.
\end{cases}    
\end{equation}

Showing the close relation between $G(mx, ux)$ and $f(u)$ is not enough. We also need to bound the difference: 
\begin{equation}
\left|\int_0^\infty G(mx, ux)^m du - \int_0^\infty f(u)^m du \right|.
\end{equation}

The steps here are the same as in \cite{Flajolet2007HyperLogLogTA}, except for the auxiliary functions (the functions that we will use to bound $f(u)$).
\begin{equation} \label{eq:25}
\begin{split}
    \biggl| \int_0^\infty G(mx,& ux)^m du - \int_0^\infty f(u)^m du \biggr| \\
    & \leq \left| \int_0^1 (G(mx, ux)^m  - f(u)^m ) du \right| \\
    & \hspace{0.5cm} + \left| \int_1^\infty (G(mx, ux)^m  - f(u)^m ) du \right|.
\end{split}
\end{equation}
We start by bounding the first term of the right-hand side of (\ref{eq:25}). We will use the easily proved fact that $1 - u/2 > f(u)$ for $u \in [0,1]$.
\begin{align*}
    \biggl| \int_0^1 (G(mx,& ux)^m  - f(u)^m )du \biggr| \\
    &\leq \left| \int_0^1 ((f(u) + u\epsilon)^m  - f(u)^m )du \right| + o(1) \\
    & \leq \left| \int_0^1 ((1 - u/2 + u\epsilon)^m  - (1 - u/2)^m )du \right| + o(1) \\
    &= \frac{4\epsilon}{1 + m} + o(1)
\end{align*}
For the second term we will use the easily proved fact that $f(u) < 1/(1+u)$ for $u \in [1, \infty)$.
\begin{align*}
    \biggl| \int_1^\infty (G(mx,& ux)^m  - f(u)^m )du \biggr|\\
    &\leq \left| \int_1^\infty ((f(u)(1 + \epsilon))^m  - f(u)^m )du \right| + o(1) \\
    & \leq \left| \int_1^\infty \left(\left(\frac{1 + \epsilon}{1 + u}\right)^m  - \left(\frac{1}{1 + u}\right)^m \right)du  \right| + o(1) \\
    &= \frac{2^{1 - m}(-1 + (1 +\epsilon)^m)}{m - 1} + o(1)
\end{align*}
The gap between the integrals completes the proof. 
\end{proof}

Using Lemma 4.3 together with (\ref{eq:13}) and (\ref{eq:7}) we concludes theorem 4.1.

\subsection{Variance Analysis}
The variance analysis helps to determine the accuracy of the algorithm. The exact expression for the second moment is
\begin{equation}
\begin{split}
    \mathbb{E}[Y^2] &= \sum_{x \in \{0, 1\}^m} \sum_{k_1,...,k_m \geq 1} \left[\left(\frac{1}{\sum_{j = 1}^m 2^{-k_j} + (1 - x_j)2^{-k_j + 1}}\right)^2 \right. \\
    & \hspace{2cm} \left.\sum_{n_1 + ...+n_m = n} \binom{n}{n_1, ...,n_m} \frac{1}{m^n} \prod_{j = 1}^m \phi_{n_j,k_j, x_j} \right],
\end{split}
\end{equation}
where $\phi_{n_j,k_j, x_j}$ is as in (\ref{eq: phi}).

Using the same steps as in the analysis of the expectation, with small modifications, we get the following asymptotic behavior as $x \rightarrow \infty$
\begin{equation} \label{eq:var}
Var_{\mathcal{P}(x)}(Y) = x^2 \left(\int_0^\infty u f(u)^m du - \left( \int_0^\infty f(u)^m du \right)^2 +\frac{\epsilon}{m} + o(1) \right),
\end{equation}
where $f(u)$ is as in theorem 4.1 and $\epsilon$ is very small.
Using Proposition 3 from \cite{Flajolet2007HyperLogLogTA} we get
\begin{equation} \label{eq:28}
    Var_n(Y) = Var_{\mathcal{P}(n)}(Y) + O(n).
\end{equation}

\section{Approximation of the Integrals}
In order to fined the constant $\gamma_m$ for the bias correction of $Y$, we have to get a closed form for the approximation of the integrals. 
\subsection{Laplace Integrals}
A Laplace integral \cite{bdab8eeda7074e4fab57b62ac70d4f6d} has the form
\begin{equation}
    I(m) = \frac{1}{m} \int_a^b \frac{f(u)}{\phi'(u)}\cdot \frac{d}{du} \left( e^{x\phi(u)}\right) du,
\end{equation}
where $m > 0$. We are interested in the asymptotic behaviour of $I(m)$ as $m \rightarrow \infty$. Thus, we will use Laplace's method \cite{bdab8eeda7074e4fab57b62ac70d4f6d}, which gives us terms from the asymptotic expansion of $I(m)$.
Denote
\begin{equation}
  \phi(u) = \log{f(u)} = \log{\left(\log_2{\frac{2 + u}{1 + u}} + \log_2{\frac{4 + 3u}{3 + 3u}} - \log_2{\frac{4 + u}{3 + u}}\right)}  
\end{equation}

Using the form of Laplace integrals we have
\begin{equation}
  I_0(m) = \int_0^\infty e^{m\phi(u)} du  
\end{equation}

\begin{equation}
  I_1(m) = \int_0^\infty ue^{m\phi(u)} du  
\end{equation}

The integral fulfills the requirements for getting the asymptotic analysis when $m$ is large enough as described in Laplace's method \cite{bdab8eeda7074e4fab57b62ac70d4f6d}. It suffices to get the first and second leading terms.
\begin{equation} \label{eq:34}
I_0(m) = \frac{\log{8}}{2m} \left(1 + \frac{1}{m}\left(\frac{41\log{2}}{16} - 1\right) + O(m^{-2})\right)
\end{equation}

\begin{equation} \label{eq:35}
I_1(m) = \left(\frac{\log{8}}{2m}\right)^2 \left(1 + \frac{3}{m}\left(\frac{41\log{2}}{16} - 1\right) + O(m^{-2})\right)
\end{equation}

\subsection{Getting the Unbiased Estimator}
In this section we will extract from the previous results the bias correction constant, which will give us the unbiased estimator. Define
\begin{equation}
\gamma_m = \frac{1}{m I_0(m)}, \hspace{1cm} \beta_m = m\left(\frac{I_1(m)}{I_0(m)^2} - 1\right).
\end{equation}
Using (\ref{eq:theorem4}) combined with the approximation that we achieved from Laplace's method (\ref{eq:34}), we get
\begin{equation}
\mathbb{E}[Y] \approx \frac{n}{\gamma_m m^2}
\end{equation}
Note that the variance of $Y$ is not the variance of the estimation. Let $\theta$ be the output of the algorithm EHLL. In the same way as we did for the expectation, using (\ref{eq:var}) and (\ref{eq:35}), we get 
\begin{equation}
\frac{Var[\theta]}{n^2} \approx \frac{\beta_m}{m}
\end{equation}
Note that when $m$ is very large we get the constants
\begin{equation} \label{eq: betta}
\gamma = \frac{2}{3\log{2}} \approx 0.962, \hspace{1cm} \beta = \frac{41\log{2}}{16} - 1 \approx 0.776.
\end{equation}

\section{Martingale Setting}
In the last decade, the martingale approach has been developed \cite{DBLP:conf/kdd/Ting14, Cohen}. The key observation was that the mergeable approach throws away all the information about the sequence of the sketches. The martingale transform can be applied to any sketch that doesn't affect by duplicates. The sketch augmented with a counter. Let $P(S)$ be the probability that the sketch will be updated. When a new element update the sketch, the counter receives the addition of $1/P(S)$.

Let $\hat{N}_t$ be the value of the counter after $t$ distinct elements arrived. Note that $\hat{N}_t - t$ is a martingale \cite{DBLP:conf/kdd/Ting14}(a martingale is a sequence of random variables, $X_1, X_2,...$, that satisfies $\mathbb{E}[|X_n|] < \infty$ and $\mathbb{E}[X_{n + 1}|X_1,...,X_n] = X_n$).  It has been shown \cite{pettie2021nonmergeable} that the martingale transform is an optimal estimator, when using a sketch in the non-mergeable way.

The martingale transform analysed by \cite{pettie2021nonmergeable} in the following way. Let $\{S_i\}_{i = 0}^n$ be a sequence of sketches, where $S_i$ is the sketch after seeing $i$ distinct elements. Denote as 
\begin{equation}
    E_n = \sum_{i = 1}^n \mathbbm{1}[S_i \neq S_{i - 1}] \cdot \frac{1}{P(S_i)},
\end{equation} 
the martingale estimation, and as 
\begin{equation}
    V_n = \sum_{i = 1}^n \mathbbm{1}[S_i \neq S_{i - 1}] \cdot \frac{1 - P(S_i)}{P^2(S_i)}
\end{equation}
the retrospective variance (it named for "retrospective", because it is an estimation of the variance). They proved that 
\begin{equation} \label{eq:41}
    \mathbb{E}[E_n] = n, \text{ and } Var(E_n) = \mathbb{E}[V_n]= \sum_{i = 1}^n \mathbb{E}\left[\frac{1}{P(S_i)}\right] - n
\end{equation}

Note that $\frac{1}{P(S_i)}$ is the indicator used in HLL and EHLL to estimate the cardinality. Thus, using (\ref{eq:41}) with the analysis of HLL \cite{Flajolet2007HyperLogLogTA, DBLP:conf/kdd/Ting14} and the analysis of EHLL (Sections 4,5), we get 
\begin{equation}
\frac{Var(E^{HLL}_n)}{n^2} =\frac{1}{n^2}\left( \frac{1}{\alpha m} \frac{n(n + 1)}{2} - n \right)\approx \frac{0.69}{m}
\end{equation}
\begin{equation} \label{eq: var_ehll}
\frac{Var(E^{EHLL}_n)}{n^2} = \frac{1}{n^2}\left(\frac{1}{\gamma m} \frac{n(n + 1)}{2} - n \right) \approx \frac{0.52}{m}
\end{equation}

\newpage

\section{Simulations}
We ran $25,000$ streams with $10^6$ distinct elements and estimated the cardinality at equal intervals. We set $m = 2^{10}$ for EHLL, Martingale EHLL and for EHLL-TC. For HLL and Martingale HLL we set $m = \left \lceil\frac{7}{6} \cdot 2^{10}\right \rceil$, and for HLL-TC $m = \frac{5}{4} \cdot 2^{10}$. So the compared sketches use the same amount of memory.

The results, described in Figure 1, match the theory and the constants (\ref{eq: betta}) and (\ref{eq: var_ehll}). It should be pointed out that EHLL needs to use $\frac{7}{6} \left(\frac{0.776}{1.082}\right) \approx 0.837$ times the space of HLL to get the same accuracy. Martingale EHLL needs $\frac{7}{6} \left(\frac{0.52}{0.69}\right) \approx 0.88$ times the space of Martingale HLL and EHLL-TC needs $\frac{5}{4} \left(\frac{0.776}{1.082}\right) \approx 0.896$ times the space of HLL-TC. 

\begin{figure*}
\centering
\begin{subfigure}[b]{0.33\textwidth}
  \centering
  \includegraphics[width=\textwidth]{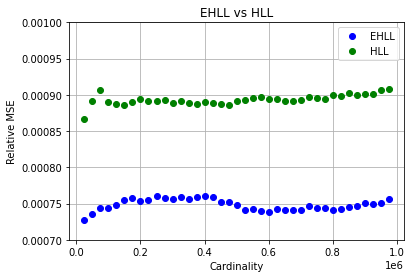}
  \caption{EHLL vs HLL}
  \label{fig:sub1}
\end{subfigure}%
\begin{subfigure}[b]{0.33\textwidth}
  \centering
  \includegraphics[width=\textwidth]{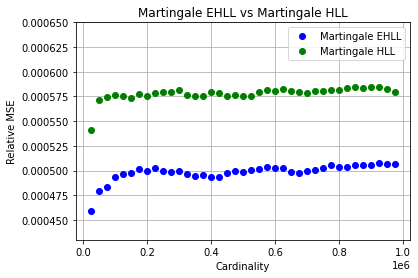}
  \caption{Martingale EHLL vs Martingale HLL}
  \label{fig:sub2}
\end{subfigure}
\begin{subfigure}[b]{0.33\textwidth}
  \centering
  \includegraphics[width=\textwidth]{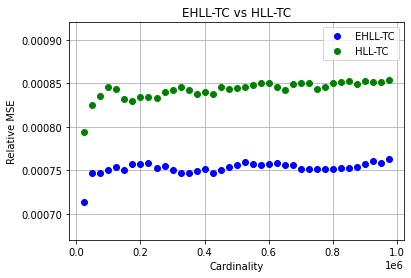}
  \caption{EHLL-TC vs HLL-TC}
  \label{fig:sub3}
\end{subfigure}
\caption{Relative MSE plots. The sketches in each graph use the same amount of memory. The variants of EHLL beat the corresponding HLL based algorithms.}
\label{fig:test}
\end{figure*}
\section{Discussion and Future Work}
We have showed here an improvement to the classic HyperLogLog algorithm. That is valid in the mergeable setting and in the sequential setting. The algorithm runs fast and gives great improvements in memory cost.

Another interesting question that arises is the space consumption against the CPU consumption. To be specific, what is the maximal number of operations per sketch update, that is still worth the space reduction. This question arises because, in many applications, the bottleneck is actually the CPU consumption.

\section*{Acknowledgement}
To Daniel Berend and Edo Liberty, for valuable discussions.

\nocite{}
\bibliography{ref}
\bibliographystyle{plain}

\appendix

\section{}
\begin{lemma} \label{bounding f(k) lemma}
Let $x > 1$, and define $f(k)$ as
\begin{equation}
f(k) = e^{-\frac{3x}{2^k}} (e^{\frac{2x}{2^{2k}}} - 1)
\end{equation}
For $k \in [1, \frac{1}{2}\log_2(2x)]$, the function $f(k)$ is monotonically increasing.
\end{lemma}

\begin{proof}
Let $a = \frac{k}{\log_2(2x)}$ so that
\begin{equation}
f(a) = e^{-\frac{3x}{(2x)^a}} (e^{\frac{2x}{(2x)^{2a}}} - 1)
\end{equation}
We will show that $f(a)$ is monotonically increasing in the interval $[\frac{1}{\log_2(2x)}, \frac{1}{2}]$. The idea is to show that $f'(a) \geq 0$ in the whole range.
\begin{align*}
    f'(a) &= e^{-\frac{3}{2} (2x)^{1 - a} + (2x)^{1 - 2a}} (-2 (2x)^{1 - 2a} \log(2x)) \\
    &\hspace{1cm} + \frac{3}{2}e^{-\frac{3}{2} (2x)^{1 - a}}(e^{(2x)^{1 - 2a}} - 1)(2x)^{1 - a} \log(2x) \\
    &= \log(2x)(2x)^{1 - a} e^{-\frac{3}{2} (2x)^{1 - a}} \left(e^{(2x)^{1 - 2a}} \left(\frac{3}{2} - 2(2x)^{-a} \right) - \frac{3}{2}\right) 
\end{align*}
The first terms are always positive. We can bound the second term using the inequality $1 +x \leq e^x$ and the fact that $(2x)^{-1/\log_2(2x)} = \frac{1}{2}$.
\begin{align*}
    e^{(2x)^{1 - 2a}} \left(\frac{3}{2} - 2(2x)^{-a} \right) - \frac{3}{2}) &\geq (1 + (2x)^{1 - 2a}) \left(\frac{3}{2} - 2(2x)^{-a} \right) - \frac{3}{2} \\
    &= (2x)^{-a}\left[(2x)^{1 - a}\left(\frac{3}{2} - 2(2x)^{-a} \right) - 2 \right] \\
    &\geq (2x)^{-a}\left(\frac{(2x)^{1 - a}}{2} - 2 \right) \\
    & \geq (2x)^{-a}\left(\frac{(2x)^{1/2}}{2} - 2 \right) > 0
\end{align*}
In the last step we used the reasonable assumption that $(2x)^{1/2} > 4$.

\end{proof}

\section{}
\begin{lemma} \label{bounding s lemma}
Let 
\begin{equation}
h(x) = \sum_{k \geq 1} \frac{x^2}{2^{2k}}e^{-\frac{3x}{2^k}} (e^{-ux2^{-k}}-e^{-3ux2^{-k}})
\end{equation}
For every $u > 0$, the function $h(x)$ has the following asymptotic behaviour as $x \rightarrow \infty$
\begin{equation}
h(x) = g(u) \left(1 + \epsilon + O\left(\frac{1}{x} \right)\right)
\end{equation}
Where $g(u) = \frac{1}{\log(2)} \left( \frac{1}{(3 + u)^2} - \frac{1}{(3 + 3u)^2} \right)$ and $|\epsilon| \leq 1.5 \cdot 10^{-5}$
\end{lemma}

The proof has the exact structure of Lemma 1 from \cite{Flajolet2007HyperLogLogTA}, but for completeness we write it for our case. 
\begin{proof}
$h(x)$ is an harmonic sum, and
\begin{equation}
h(x) = \sum_{k = 1}^{\infty} q(\mu_k x)
\end{equation}
where $q(t) = t^2(e^{-t(3 + u)} - e^{-t(3 + 3u)})$ and $\mu_k = \frac{1}{2^k}$.
Thus, its Mellin transform factorizes, in the fundamental strip $\langle -3, 0 \rangle$, as
\begin{equation}
h^{\star}(s) = \left (\sum_{k = 1}^{\infty} 2^{ks} \right)q^{\star}(s) = \frac{2^s \Gamma(s + 2)}{1 - 2^s} \left( \frac{1}{(3u + 3)^{s + 2}} - \frac{1}{(u + 3)^{s + 2}}\right).
\end{equation}
By Mellin inversion theorem, we have that
\begin{equation}
h(x) = \frac{1}{2 i \pi} \int_{c - i \infty}^{c + i \infty} h^{\star}(s) x^{-s} ds
\end{equation}
for every $c \in \langle -3, 0 \rangle$. $h^{\star}(s)$ has poles at $\mathbb{Z}_{<0}$, because of the $\Gamma$ function, and at the complex values $\{\eta_k = \frac{2 i k \pi}{\log{2}}, k \in \mathbb{Z}\}$. Using the residue theorem we get
\begin{equation}
h(x) = \sum_{k \in \mathbb{Z}}\text{Res}(h^{\star}(s) x^{-s}, \eta_k) - \frac{1}{2i \pi} \int_{1 - i \infty}^{1 + i \infty} h^{\star}(s) x^{-s} ds.
\end{equation}

The residues are
\begin{equation}
\begin{cases}
    \text{Res}(h^{\star}(s) x^{-s}, \eta_k) = \frac{1}{\log{2}} x^{-\eta_k} \Gamma(\eta_k) ((3 + u)^{-\eta_k - 2}\\  \hspace{4cm}- (3 + 3u)^{-\eta_k - 2}),& k \in \mathbb{Z}_{\neq 0}, \\
    \text{Res}(h^{\star}(s) x^{-s}, 0) = \frac{1}{\log{2}} \left( \frac{1}{(3 + u)^2} - \frac{1}{(3 + 3u)^2} \right),& k = 0.
\end{cases}
\end{equation}
We will use the following inequality. When $\Re(s) \geq 0$
\begin{equation} \label{eq: 53}
    |(3 + u)^{-s-2} - (3 + 3u)^{-s-2}| \leq |s + 2| \left(\frac{1}{(3 + u)^2} - \frac{1}{(3 + 3u)^2} \right).
\end{equation}
The inequality (\ref{eq: 53}) can be verified in the same way as suggested in \cite{Flajolet2007HyperLogLogTA} (we wrote the differences as the integral of the derivative, with the relevant bounds, and bound the derivative).
And the fact that 
\begin{equation}
\sum_{k \in \mathbb{Z}_{\neq 0}} |(k + 2) \Gamma(k)| \leq 1.65 \cdot 10^{-6}
\end{equation}
to get 
\begin{equation}
\left|\sum_{k \in \mathbb{Z}_{\neq 0}}\text{Res}(h^{\star}(s) x^{-s}, \eta_k)\right| \leq
 \frac{\epsilon}{\log{2}}  \left(\frac{1}{(3 + u)^2} - \frac{1}{(3 + 3u)^2} \right)
\end{equation}
where $\epsilon = 1.5 \cdot 10^{-5}$.
With the bound of the integral
\begin{align*}
\left|\int_{1 - i \infty}^{1 + i \infty} h^{\star}(s) x^{-s} ds \right| &\leq 
\frac{1}{x} \frac{1}{\log{2}}\left(\frac{1}{(3 + u)^2} - \frac{1}{(3 + 3u)^2} \right) \\
& \hspace{1cm} \cdot \int_{1 - i \infty}^{1 + i \infty}|(s + 2) \Gamma(s + 2)| d|s| \\
&= O\left(\frac{1}{x} \right) \frac{1}{\log{2}} \left(\frac{1}{(3 + u)^2} - \frac{1}{(3 + 3u)^2} \right)
\end{align*}
we complete the proof.

\end{proof}

\end{document}